\documentclass[journal]{IEEEtran}
\usepackage{amssymb}
\usepackage{amsmath}
\usepackage{amsfonts}
\usepackage{graphicx}
\usepackage{algorithm}
\usepackage{algorithmic}
\usepackage{epstopdf}
\usepackage{cite}
\usepackage{exscale}
\usepackage{relsize}
\usepackage[
top    = 1.70cm,
bottom = 1.05in,
left   = 0.65 in,
right  = 0.65 in]{geometry}

\usepackage{color}
\usepackage{multirow}
\usepackage{subfigure}
\usepackage{bm}

\usepackage{hyperref}

\newtheorem{lemma}{\textbf{Lemma}}

\begin{document}
\title{Modeling and Analysis for Cache-enabled Cognitive D2D Communications in  Cellular Networks}
\author{Chenchen Yang$^{*}$, Xingya Zhao$^{*}$,  Yao Yao$^{\ddag}$, Bin Xia$^{*}$\\
$^{*}$ Department of Electronic Engineering, Shanghai Jiao Tong University, Shanghai, P. R. China\\
$^{\ddag}$ Huawei Technologies Co., Ltd\\
Email: {\{zhanchifeixiang, claret.zhao, bxia\}@sjtu.edu.cn, yyao@eee.hku.hk}
\thanks{This work is supported in part by the National High Technology
Research and Development Program of China under 863 5G (Grant NO.
2014AA01A702), and the Shenzhen-Hong Kong Innovative Technology Cooperation
Funding (Grant No. SGLH20131009154139588).}}
\maketitle
\begin{abstract}
Exploiting cognition to the cache-enabled device-to-device (D2D) communication underlaying the multi-channel cellular network is the main focus of this paper. D2D pairs perform direct communications via sensing the available cellular channels, bypassing the base station (BS). Dynamic service is considered and the network performance is evaluated with the stochastic geometry. Node locations are first modeled as
mutually independent Poisson Point Processes, and the service queueing process is formulated. Then the corresponding tier
association and cognitive access protocol are developed. The delay and the length for the queue at the BS and D2D transmitter are further elaborated, with modeling the traffic dynamics of request arrivals and departures as the \emph{discrete-time multiserver queue with priorities}. Moreover, impacts of the physical layer and content-centric features on the system performance are jointly investigated  to provide a valuable insight.
\end{abstract}
\section{INTRODUCTION}
As an essential candidate feature of the fifth generation (5G) network, device-to-device (D2D) communications is an excellent technology for enhancing the cellular coverage and capacity \cite{jsacd2d}. Employing D2D communication in cellular networks with the concepts of cognitive radio networks is one of the schemes that are being intensively studied for future networks \cite{cognitive2, cognitive1}. On the other hand, despite the totally huge amount of data traffic suffering by the network, only a small portion $(5 - 10\%)$  of  ``popular"  data are frequently accessed by the majority of the users \cite{Zipf}. Caching popular contents at terminal devices allows frequency local sharing via D2D, yielding  the traffic offloading from the cellular network  and the reduction of duplicated data transmissions \cite{diversity}.

D2D communications in cellular networks are either network-assisted or assistance-free \cite{jsacd2d}. In the former, the base stations (BS) takes centralized control to allocate the resources for both cellular and D2D communications, which  may incur dramatically high computational and signaling overhead, especially in the high-density network with large numbers of BSs,  cellular and D2D users \cite{jsacd2d, underlay1}. On the other hand, the assistance-free D2D advocates to manage the communication with distributed D2D terminals bypassing BSs. Compared with centralized approaches, distributed resource allocation schemes incorporating the concept of cognitive networks for autonomous resource allocation will be more scalable \cite{jsacd2d, cognitive1}.

Cognition is exploited for the assistance-free communication in \cite{cite1}. The D2D transmitter (TX) and receiver set up direct communication without the supervision from the BS via the PC5 interface, which is defined by the  Long Term Evolution (LTE) standard. \cite{jsacd2d}  focus on mixed overlay-underlay cognitive  D2D communications where devices sense the spectrum based on the energy detection. In \cite{harvest}, the cognitive D2D TXs can communicate with the harvest energy from the ambient interference via the cellular channel. Both the random spectrum access and prioritized spectrum access policies are investigated. Authors in \cite{CSMA} introduce cognition to femto access points to exploit the idle channels of the macro BSs via spectrum sensing, avoiding severe interference. \cite{cog2} introduces the spectrum trading scheme to the cognitive D2D and cellular network coexisting system based on game theory and learning method.

Cache-enabled D2D transmission implies that data can be cached at the terminal devices for content sharing via D2D \cite{5G2}. \cite{scaling} identifies the connection between the collaboration distance and the interference in a cache-enabled D2D network with considering the video content popularity statistics. \cite{grid} divides the cell into virtual square grids to make use of the caching capacity of terminal devices, exploiting the redundancy of user requests. The performance of different cluster centric content placement policies in a cache-enabled D2D network is analyzed in \cite{arxiv2015}, where the location of devices are modeled as a Poisson cluster process. In \cite{trade-off}, disjoint circular clusters are set based on hard-core point process. Requesting users access contents from cache-enabled users in the same cluster via out-of-band D2D in the cellular network.

Nevertheless, compared to the network-assisted D2D, assistance-free D2D underlaying cellular networks has not been investigated enough. Moreover, few studies exploit cognition to the cache-enabled D2D communications. Modeling and exploring the performance of the cellular network embedded with the cognitive and cache-enabled D2D communication is the focus of this paper. The main contributions of this paper are summarized as follows:
\begin{itemize}
\item \emph{Introducing cognition to the cache-enabled D2D underlaying cellular networks:}
  Cache-enabled users (active as D2D TXs) and BSs are cooperative to transmit contents. Assistance-free D2D communication is considered and D2D pairs perform direct communications via spectrum sensing in the multi-channel cellular environment.
\item \emph{Dynamic service analysis with the stochastic geometry:}
 We use tools of  the stochastic geometry to model the node locations. Different from the traditional assumption in the stochastic geometric analysis that service nodes are active at the full-load state, we consider the dynamic service and the request queueing process is formulated.
\item \emph{Evaluating the performance and facilitating the straightforward network configuration:}
  The probability mass function (PMF) of the delay and the length for the queue at the BS and the D2D TX are derived, modeling the traffic dynamics of request arrivals and departures at the BS and D2D TX  as the \emph{discrete-time multiserver queue with priorities}.
\item \emph{Highlighting the joint impacts of  the physical layer and content-centric features:} Impacts of the content-centric features (e.g., content popularity, the request intensity and the limited caching storage)  and physical layer parameters (e.g., node densities, channel fading and transmission powers) on the network performance are investigated.
\end{itemize}
\section{System model and Protocol description}\label{sec:system}
In this section, the network architecture is first elaborated. Then the tier association protocol and the cognitive channel access are further developed.
\subsection{Network Architecture}
We consider a network as illustrated in Fig. \ref{structure}, where BSs and users are spatially distributed according to two mutually independent homogeneous PPPs $\Phi_i= \{b_{i,j}, j=0,1,2,...\}$ with intensity ${\lambda}_i$ ($i = 0,2$ for users and BSs, respectively) \cite{ppp2}, and $\lambda_0\gg\lambda_2$. A proportion ($0<\alpha<1$) of users are cache-enabled and they can be active as D2D TXs, yielding a thinning homogeneous PPP $\Phi_1= \{b_{1,j}, j=0,1,2,...\}$ with intensity ${\lambda}_1=\alpha{\lambda}_0$.

Users request contents from a given library $\{1,2,...,N\}$ and all contents are of equal length $L$ $[$bits$]$.  Every cache-enabled user has cache of the same size  $M \times L$ $[$bits$]$, where $M \ll N$.  The request interarrival times of a user are exponentially distributed random variables with mean $\frac{1}{\lambda_u}$ seconds \cite{Queue}, i.e., requests of a single user is a Poisson process with parameter $\lambda_u$ $[$requests/s$]$.
 The requested probability of the $i$-th content is triggered according to the Zipf distribution \cite{Zipf,scaling} as follow
\begin{equation}\label{zipf}
f_{i}={\frac{1/i^{\nu}}{\sum_{j=1}^{N}1/j^{\nu}}},
\end{equation}
where $\nu \geq 0$ reflects the difference of the access frequency to different contents. Large $\nu$ implies that the access to a few of  popular contents accounts for the majority of the total access traffic. The content with lower ranking (larger $i$) is more rarely accessed by users.
\begin{figure}[t]
\centering
\includegraphics[width=2.5in]{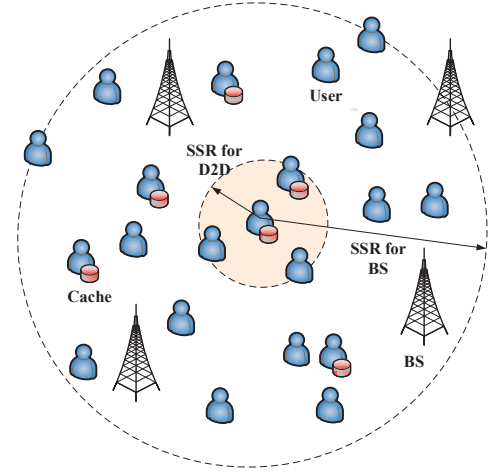}
\caption{The cognitive and cache-enabled D2D transmission.}
\label{structure}
\end{figure}
\subsection{Tier Association  Protocol}\label{subsec:protocol}
When the networks is off-peak, e.g., in the nighttime, BSs proactively broadcast  the most popular contents  to the cache-enabled users.  All the cache-enabled users pre-cached the same copy of the pushed contents, i.e. the $M$  most popular contents. The tier association protocol is jointly decided by the following physical layer and the content-centric features.
\subsubsection{\textbf{Physical layer features}} We consider the node providing the highest long-term average received signal strength to a requesting user as the ``closest'' node of the user. And the received signal strength is  \cite{ppp2},
 \begin{equation}
T_{i}=\kappa B_iP_ir_i^{-\beta}, ~\text{for}~ i=1,2,
 \end{equation}
Where $\beta\geq 2$ is the path-loss exponent and $r_i$ denotes the distance between the requesting user and its nearest node in the $i$-th tier. $P_i$ for $i=1,2$ are the transmit power of a D2D TX and a BS, respectively. For clarity, both the association bias factor $B_{i}$ and  the propagation constant $\kappa$ are normalized as $1$ in this paper. So the ``closest'' node is $\arg\max_{i}T_{i}$.
\subsubsection{\textbf{Content-centric features}} If a request is triggered by a user who is not cache-enabled,  the user associates to its ``closest'' node\footnote{Note that when its ``closest'' node is a D2D TX, if the D2D TX has cached the requested content, then a \emph{cache-hit} event happens  and the user obtains the content via the D2D TX; otherwise, the \emph{cache-loss} event happens because all the other D2D TXs  have not  cached the requested content either, and the user has to skip over D2D TXs to associate to the BS tier.}.  If the requesting user is cache-enabled but has not cached the requested content, the user accesses to its ``closest'' BS.
If the requesting user is cache-enabled and the requested content has been cached in its local caching storage, the user obtains the content immediately.

Therefore,  in a snapshot of the network, requesting users can be clarified into the following three non-overlapping Subsets according to where (its local caching, the D2D TX or the BS) it obtains the requested content.

\textbf{Subset 0:} The user who can obtain the requested content from its local caching immediately.

\textbf{Subset 1:} The user who can  obtain the requested content from the D2D TXs. It implies that the requested content has been cached in the user's ``closest'' D2D TXs.

\textbf{Subset 2:} The user who can  obtain the requested content from the BS. It includes two possible cases, the first case is that the BS is the ``closest'' server of the user. The other case is that the user is not cache-enabled, and its ``closest'' server is a D2D TX who yet has not cached the requested content (which means the requested content has not cached in the D2D tier). Then the user has to access the ``closest'' BS.

The cellular network without caching is considered as a baseline in this paper. In the baseline,
 users do not have caching ability and they could not pre-cache the popular contents. The local sharing links via D2D communications among users can not work at this time. All the users need to access the BS for the content.
\subsection{Priority Based Cognitive Channel Access}
A set of orthogonal channels $C = \{c_i, i=1, 2, ...,|C|\}$ are available in the downlink.  Different BSs can reuse all the channels and no channel is reused within the same cell. Time is divided into slots with fixed-length interval $\tau$ for the BS and  D2D TXs to sent contents, we consider $\tau=1$ second for analytical tractability. A BS assigns one channel to every received request at the beginning of each time slot. The requests that arrive to the BS in the middle of a time slot will be assigned channels until the next time slot. Each BS assigns channels to received requests from $c_1$ to $c_{|C|}$ in sequence, i.e.,  channel $c_{i+1}$ will not be assigned before channel $c_i$.
\begin{figure}[t]
\centering
\includegraphics[width=3in]{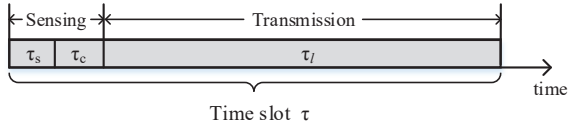}
\caption{Time allocation of a slot for the D2D sensing and transmission.}
\label{slot}
\end{figure}

The cache-enabled users take the free channels which are not used by the nearby BSs and the nearby cache-enabled users, to share the local cached contents to the requesting users via D2D.  Cache-enabled users are cognitive and access the available channels around them opportunistically \cite{jsacd2d, cognitive1, underlay1, cite1}. A channel is available to a cache-enabled user if and only if its received power on this channel, whether provided by a BS node or a cache-enabled user, is less than the spectrum sensing threshold $\gamma$. For a specific cache-enabled user, the spectrum sensing threshold determines its SSR, in which the received power from any other transmit node is larger than $\gamma$, yielding the non-reusable of the corresponding  channel. The larger the value of the spectrum sensing threshold $\gamma$, the smaller the SSR area, and vice versa. Because the transmit power of the BS is larger than that of the D2D TX, the average radius of the SSR for BSs are larger than that for D2D TXs as in Fig. \ref{structure} (we shall note that the instantaneous SSRs are random shapes because of the channel fading). Then each cache-enabled user serves its received requests according to the following Steps in a slot as illustrated in Fig. \ref{slot} \cite{CSMA}:

\textbf{Step 1:} The cache-enabled user senses the spectrum during the sensing time ${\tau}_s\ll\tau$ and determines the free channels not used by the BSs within its SSR with regard to BSs (called BS SSR in the following);

\textbf{Step 2:} Select a number of channels from the available channels sensed in Step 1. Persistently sense each of them for a random duration $t$ uniformly distributed in the region $[0,\tau_c]$, where $\tau_c\ll\tau$. If there is still any channel available after the sensing duration (i.e., no nearby cache-enabled user accesses the channel during $t \in [0,\tau_c]$ ), go to Step 3; otherwise, delay the content transmission to the next time slot.

\textbf{Step 3:} During the interval $\tau_l$, D2D TXs access the available channels for D2D transmission to the firstly unresponded requests. Note that the D2D TX and receiver can perform direct communications over the PC5 interface defined by the LTE standard, bypassing the supervision of the BS \cite{jsacd2d, cite1}.

We consider perfect spectrum sensing  and the same spectrum access procedure is repeated in each time slot \cite{CSMA}.
For a cache-enabled user, how many channels it tries to occupy depends on the number of received requests it has to respond in this time slot. These requests include the newly received ones in the duration of the last time slot, and the ones which are received in more previous slots but have not been responded due to the shortage of available channels.
\section{PROBLEM FORMULATION AND ANALYSIS}
In this section, we will derive the tier association probabilities for a typical user. Then the number of users associated to each BS are analyzed. The delay and the queue length at the BS and the D2D TX are further evaluated with the modeling of the \emph{discrete-time multiserver queue with priorities}.
\subsection{Tier association probability}
We have modeled a two-tier network above, where the cellular network is overlaid with D2D tier. The probability that the average received signal strength from the $i$-th tier of a user is higher than that from the $j$-th tier is \cite{ppp2,twcmy}
\begin{align}
 \mathrm{\mathbb{P}(T_{i}>T_{j})}=\left[{\sum\limits_{k=1}^2\frac{\lambda_k}{\lambda_i}\left(\frac{P_k}{P_i}\right)^{\frac{2}{\beta}}}\right]^{-1}, i\neq j\in\{1,2\}.\label{probability_2}
\end{align}	
For a randomly selected requesting user, which Subset it belongs to is jointly decided by the physical layer and content-centric features. In a specific time slot, the portions $\mathbb{P}_{ui}$ for the requesting users in the corresponding Subsets are \cite{twcmy}
\begin{equation}
 \left\{
\begin{array}{ll}
\mathbb{P}_{u0}=\alpha\mathbb{P}_{h} \\
\mathbb{P}_{u1}=(1-\alpha)\mathbb{P}_{h}\mathbb{P}(T_1>T_2)\\
\mathbb{P}_{u2}=(1-\mathbb{P}_{h}) + (1-\alpha)\mathbb{P}_{h}[1-\mathbb{P}(T_1>T_2)],
\end{array}
 \right.
\end{equation}
where $\mathbb{P}_{h}=\sum_{i=1}^Mf_i$ is the probability that the \emph{cache-hit} event happens, i.e., the requested content has been cached in the caching storage. $\mathbb{P}_{ui}$ for $i=0,1,2$ also can be viewed as the probability that a randomly selected requesting user can obtain its requested content from its local caching, the D2D TX and the BS, respectively.
\subsection{The Number of Users Associated to each BS}
As explained above, we have divided users into three non-overlapping Subsets in each time slot. With marking users based on which Subset they belong to, we approximately get three  independently thinning PPPs $\Phi_{ui}$ with intensity $\lambda_{ui}=\mathbb{P}_{ui}\lambda_0$ for $i=0,1,2$ in the $\mathbb{R}^2$ plane \cite{CSMA}. That is, the locations of users who obtain their requested contents from local caching, D2D TXs and BSs  in each slot are spatially distributed according to mutually independent PPPs with intensity  $\lambda_{ui}$ for $i=0,1,2,$ respectively.

Therefore,  the number of users associated to a BS can be formulated within the homogeneous BS tier, where the BS cells are polygonal and form the Voronoi tessellation in $\mathbb{R}^2$. The size of the Voronoi cell area is a random variable, and its probability density function (PDF) can be accurately predicted by the gamma distribution \cite{CSMA}
 \begin{equation}\label{ps}
f_S(s)=\frac{(\lambda_2 K)^Ks^{K-1}e^{-\lambda_2Ks}}{\Gamma(K)}, 0\leq s<\infty,
 \end{equation}
where $s$ denotes the cell size and $\Gamma(\cdot)$ is the gamma function. And  $K=3.575$ is a constant factor.
Then with the property of the PPP, the number of users $N_2$ in a BS cell  conditioning on the cell size is a Poisson random variable, which can be generated with the conditioned PMF
\begin{equation}\label{pconditon}
\mathbb{P}\left(\begin{array}{c|}
\!\!\!N_2=n\!
\end{array}~S=s\right)\!=\!\frac{(\lambda_{u2}s)^n\text{exp}\{\!-\!\lambda_{u2}s\}}{n!}, n=1,2,...,
\end{equation}
we then have the following lemma  \cite{CSMA}.
\begin{lemma}
The PMF of the number of users associated to a BS is
\begin{align}
&\mathbb{P}(N_2=n)=\frac{\lambda_{u2}^n(K\lambda_2)^K}{(\lambda_{u2}+K\lambda_2)^{K+n}}\frac{\Gamma(K+n)}{\Gamma(n+1)\Gamma(K)}.
\end{align}
\end{lemma}
\begin{proof}
By considering (\ref{ps}) and (\ref{pconditon}), we have
$\mathbb{P}(N_2=n)=\int_{0}^\infty\mathbb{P}\left(\begin{array}{c|}
                      N_2=n
                   \end{array}~S=s\right)f_{S}(s)\mathrm{d}{s}$, we obtain the lemma.
\end{proof}
\subsection{Number of BSs and D2D TXs  in the SSR}\label{sec:number}
The spectrum sensing threshold $\gamma$ defines the SSR around a reference D2D TX in the two-tier heterogeneous network (the D2D tier and the BS tier).  In the SSR of a reference D2D TX, the number of D2D TXs (the reference TX is not included) and BSs $N_{\gamma i},i=1,2$  can be described as
\begin{equation}
N_{{\gamma}i}=\sum_{b_{i,j}\in\Phi_i}\mathbf{1}\left\{P_i h_{i,j}r_{i,j}^{-\beta}>\gamma\right\},
\end{equation}
where $r_{i,j}$ is the distance from the node $b_{i,j} \in \Phi_i$ for $i=1,2$ to the reference D2D TX. We consider Rayleigh fading channels. The channel gains from the $j$-th node in the $i$-th tier to the reference D2D TX are  $h_{i,j}\sim\text{Exp}(\mu)$, and they are independent with each other. We then have the following lemma \cite{CSMA}.
\begin{lemma}
The number of D2D TXs and BSs $N_{\gamma i}$ ($i=1,2$)  in the SSR of the reference D2D TX follow the Poisson distribution with intensity $\lambda_{\gamma i} = \pi\lambda_i(\frac{P_i}{\gamma\mu})^{\frac{2}{\beta}}\Gamma(1+\frac{2}{\beta})$, respectively.
\end{lemma}
\begin{proof}
The moment generation function (MGF) of $N_{\gamma i}$ is
\begin{align}\label{eet}
\mathbb{E}(e^{tN_{{\gamma}i}})&=\mathbb{E}\Big(e^{t\sum_{b_{i,j}\in\Phi_i}\mathbf{1}\left\{P_i h_{i,j}r_j^{-\beta}>\gamma\right\}}\Big)\nonumber\\
&=\mathbb{E}_{\Phi_i}\Big[\prod\nolimits_{b_{i,j}\in\Phi_i}\mathbb{E}_{h_{i,j}}e^{t\mathbf{1}\left\{P_i h_{i,j}r_j^{-\beta}>\gamma\right\}}\Big]\nonumber\\
&\stackrel{(a)}{=}\text{exp}\Big\{-\mathbb{E}_{h_{i,j}}\!\Big[\!\int_0^{2\pi}\!\!\!\!\int_0^{\left(\frac{P_ih_{i,j}}{\gamma}\right)^{\frac{1}{\beta}}}\!\!\!\!\!(1-e^t)\lambda_ir\mathrm{d}{r}\mathrm{d}{\theta}\Big]\!\Big\}\nonumber\\
&\stackrel{(b)}{=}\text{exp}\Big\{-\pi(1-e^t)\lambda_i\mathbb{E}_{h_{i,j}}\Big[{\Big(\frac{P_ih_{i,j}}{\gamma}\Big)^{\frac{2}{\beta}}}\Big]\Big\}\nonumber\\
&\stackrel{(c)}{=}\text{exp}\Big\{-\pi(1-e^t)\lambda_i\Big(\frac{P_i}{ \gamma\mu}\Big)^{\frac{2}{\beta}}\Gamma(1+\frac{2}{\beta})\Big\},
\end{align}
where Step (a) has been proven in \cite{2009interference,CSMA}. Step (b) is obtained by noting that $h_{i,j}\sim\text{Exp}(\mu)$. Step (c) is obtained based on
\begin{align}
\mathbb{E}_{h_{i,j}}\Big[{\Big(\frac{P_ih_{i,j}}{\gamma}\Big)^{\frac{2}{\beta}}}\Big]&=\int_0^{+\infty}\mu e^{-\mu x}\Big(\frac{P_ix}{\gamma}\Big)^{\frac{2}{\beta}}\mathrm{d}{x}\nonumber\\
&=\Big(\frac{P_i}{\gamma\mu}\Big)^{\frac{2}{\beta}}\Gamma\Big(1+\frac{2}{\beta}\Big).
\end{align}
Compared (\ref{eet}) with the MGF of the Poisson distribution, we can find that $N_{\gamma i}$ is distributed according to Poisson process with intensity $\lambda_{\gamma i} = \pi\lambda_i(\frac{P_i}{\gamma\mu})^{\frac{2}{\beta}}\Gamma(1+\frac{2}{\beta})$.
\end{proof}

\subsection{The Delay and the Queue Length}
In this part,  we analyze the queue length of requests at the BS and the D2D TX, as well as the delay of users in different Subsets. Requests that arrive at the same node (a BS or a D2D TX) are queued in a infinite buffer until they can be served, and they are served on a FIFO-basis (first-in, first-out). Consider the strategy that a request will be deleted from the buffer no matter whether the corresponding content has been transmitted completely in a slot scheduled to it (at the cost of increasing the content loss rate), i.e., without repeat transmission.
\subsubsection{\textbf{The delay and the queue length at the BS}}
As to any given BS node, the arrival of received requests follows a Poisson process with intensity $n_2\lambda_u$, conditioned on  the number $N_2=n_2$ of associated users.   Note that the BS have the priority to assigning channels to its received requests at the beginning of each time slot, which is not affected by the state of any D2D TX. The queueing behavior of BSs with high priority is not affected by the presence of D2D TXs with low priority. Then traffic dynamics of request arrivals and departures at a BS is a \emph{discrete-time multiserver queue} (with $|C|$ servers).  Substituting the Poisson arriving process into the result of the reference \cite{bsqueue}, we get the conditional probability generating function (PGF) of the queue length at the BS (note that the conditional PGF here is for the BS such that the queue at the BS is steady, i.e., $n_2\lambda_u<|C|$),
\begin{align}
\mathcal{L}_b(z|N_2=n_2)&= \frac{e^{n_2\lambda_u(z-1)}(|C|-n_2\lambda_u)(z-1)}{z^{|C|}-e^{n_2\lambda_u(z-1)}}\nonumber\\
&\times\prod_{i=1}^{|C|-1} \frac{z-\beta_i}{1-\beta_i},
\end{align}
where the $\beta_i$ ($i = 1,2,...,|C|-1)$ are the $|C|-1$ solutions of $z^{|C|}=e^{n_2\lambda_u(z-1)}$ inside the complex unit disk minus the point $z=1$. We refer readers to \cite{bsqueue} for detail derivations. Then the conditional PMF of the queue length $L_b$ at the BS can be given with different methods, e.g., the  inverse discrete Fourier transform (IDFT) method \cite{bsqueue},
$
\mathbb{P}\left(L_b=n|N_2=n_2\right)\approx\frac{1}{W}\sum\limits_{w=0}^{W-1}\frac{\mathcal{L}_b\left(e^{j2\pi w/W}|N_2=n_2\right)}{e^{j2\pi wn/W}}.
$
This method will yield more accurate approximations for the PMF when choosing larger $W$. Therefore, the PMF of the queue length $L_b$ at the BS is
\begin{align}
\mathbb{P}(L_b=n)=\sum\limits_{n_2=0}^{\lceil{\frac{|C|}{\lambda_u}}\rceil-1}\mathbb{P}(L_b=n|N_2=n_2)\mathbb{P}(N_2=n_2),
\end{align}
where $\lceil x\rceil$ is the the ceiling function. Similarly, the conditional PGF of the delay in the queue at a BS is given by
$\mathcal{D}_b(z|N_2=n_2) = \sum_{k=0}^{|C|-1} \frac{(v_k^{|C|}(z)-1)v_k(z)}{|C|\cdot(v_k(z)-1)} \times \frac{(|C|-n_2\lambda_u)(e^{n_2\lambda_u(v_k(z)-1)}-1)}{n_2\lambda_u(v_k^{|C|}(z)- e^{n_2\lambda_u(v_k(z)-1)}  )} \prod_{i=1}^{|C|-1} \frac{v_k(z)-\beta_i}{1-\beta_i},$
whereby $v_k(z)$ are the $|C|$ different solutions of $v^{|C|}=z$. Then the PMF of the delay $\mathbb{P}(D_b=n)$  can be further obtained.
\subsubsection{\textbf{The delay and the queue length at the D2D TX}}
Due to the large gap between the transmit power of BSs and D2D TXs, the area of the D2D SSR is much smaller than that of the BS SSR. It is reasonable to assume that all D2D TXs coexisting within one D2D SSR share the same BS SSR and the same free channels for tractable analysis \cite{CSMA}. Regard a D2D TX and the other D2D TXs in its D2D SSR as a group, the free channels in the D2D SSR are shared by the group orthogonally. According to the analysis in Sec. \ref{sec:number}, we can get the number of users associated to the group of D2D TXs follows the Poisson distribution with intensity  $\lambda_{\gamma u}=\pi \mathbb{P}_{u1}\lambda_0(\frac{P_1}{\gamma\mu})^\frac{2}{\beta}\Gamma(1+\frac{2}{\beta}) $. Then the PMF $\mathbb{P}(N_{\gamma u}=n)$   of the number of users  associated to  the group of  D2D TXs can be obtained.


For the priority of the BSs, the number of free channels available for the D2D group is the number of channels that are not used by the BS with most received requests to be responded to in the D2D SSR. As to each BS, the requests it received during each time slot follows a Poisson process as aforementioned. We consider the BS with most users associated with it to be the one that occupies most channels. Denote $N_2^j$ as the number of users associated with the $j$-th BS. The conditional CDF of the number of users under the BS with the heaviest load  is $N_{K}$
\begin{align}
\mathbb{P}({N_{K}\leq k|N_{\gamma 2}=n})
&= \mathbb{P} \Big({\max_{j=1,2,...,n}N_{2}^j\leq k}\Big) \nonumber\\
&= \Big[\sum_{i=0}^k \mathbb{P}(N_2=i)\Big]^n.
\end{align}
Then the CDF of the number of users under the BS with the heaviest load is
\begin{align}
\mathbb{P}({N_{K} \leq k})
&=\sum_{n=0}^{+\infty} \mathbb{P}({N_{\gamma2}=n})\mathbb{P}({N_K\leq k|N_{\gamma 2}=n}).
\end{align}

Then the system can be regarded as a two-priority queuing system with $|C|$ servers, i.e., a  \emph{discrete-time multiserver queue with priorities} \cite{1998discrete}. The requests with high-priority (associated to the BS with the heaviest load) arrive according to a Poisson process with intensity $\lambda_H = N_K\lambda_u$. The arrival of requests with low-priority (associated to the group of D2D TXs) follows a Poisson process with intensity $\lambda_L = N_{\gamma_{u}}\lambda_u$. The two processes are independent to each other, but the high-priority requests may lead to server interruptions to the queueing of the low-priority requests.

We provide an  insight into the performance of the D2D tier with the state of the low-priority queueing. Substituting the Poisson process into results of \cite{1998discrete} which is with regard to a general arriving process, we get the conditional PGF for the queue length of requests with low priority
\begin{align}
\mathcal{L}_d(z|N_K, N_{\gamma u})&= \frac{e^{\lambda_L(z-1)}(|C|-\lambda_T)(z-1)}{1-e^{\lambda_L(z-1)}}\nonumber\\
&\times\prod_{i=0}^{|C|-1} \frac{1-x_i(z)}{z-x_i(z)} \prod_{i=1}^{|C|-1}\frac{z-\alpha_i}{1-\alpha_i}.
\end{align}
whereby the $\alpha_i$ ($i = 1,2,...,|C|-1$) are the $|C|-1$ solutions of $z^{|C|}=e^{\lambda_T(z-1)}$ inside the complex unit disk minus the point $z=1$, and we have $\lambda_T=\lambda_H+\lambda_L$. The functions $x_i(z)$ are the $|C|$ solutions for equation $x_i^{|C|}(z) = e^{\lambda_H(x_i(z)-1)+\lambda_L(z-1)}$ for any given $z$. Moreover, the conditional PGF of the delay for low-priority requests is given by
$
\mathcal{D}_d(z|N_K, N_{\gamma u}) = \sum_{i=0}^{|C|-1}r_i(z)Q_T(\omega_i(z)),
$
where $
r_i(z|N_K, N_{\gamma u}) = z \prod_{j=0,j\neq i}^{|C|-1} \frac{1-\omega_j(z)}{\omega_i(z)-\omega_j(z)}, 0 \leq i <|C|,
$ and $
Q_T(z|N_K, N_{\gamma u}) = \frac{(|C|-\lambda_T)\left(e^{\lambda_T(z-1)}-e^{\lambda_H(z-1)}\right)}{\lambda_L\left(z^{|C|}-e^{\lambda_T(z-1)}\right)}
\times\prod_{i=1}^{|C|-1} \frac{z-\alpha_i}{1-\alpha_i},
$
and $\omega_i(z)$ are the solutions of $\omega^{|C|} = ze^{\lambda_H(\omega-1)}$. Similarly, the PMF of the queue length and the delay can be obtained with the  IDFT method described above.

\section{Numerical Results}
We verify the performance of the proposed system with simulating the cache-enabled network. The results are obtained with Monte Carlo methods in a square area of $2000m\times2000m$, where the locations of users and BSs are independent homogeneous PPPs with intensities of $\{\lambda_0, \lambda_2\} = \{\frac{100}{\pi500^2},\frac{1}{\pi500^2}\}$ nodes/m$^2$. We set the path-loss $\beta = 4$, the transmit powers are $\{P_0,P_2\}=\{23,43\}$ dBm, the request intensity $\lambda_u=0.09$ [requests/s],  the number of channels $|C|=10$, total number of contents $N = 200$, the caching ability $M = 10$, and the content popularity $\gamma = 1$. These typical parameters will not change unless specific statements are clarified.
\begin{figure}
  \centering
  \subfigure[]{
    \label{fig:subfig:a} 
    \includegraphics[width=1.8in]{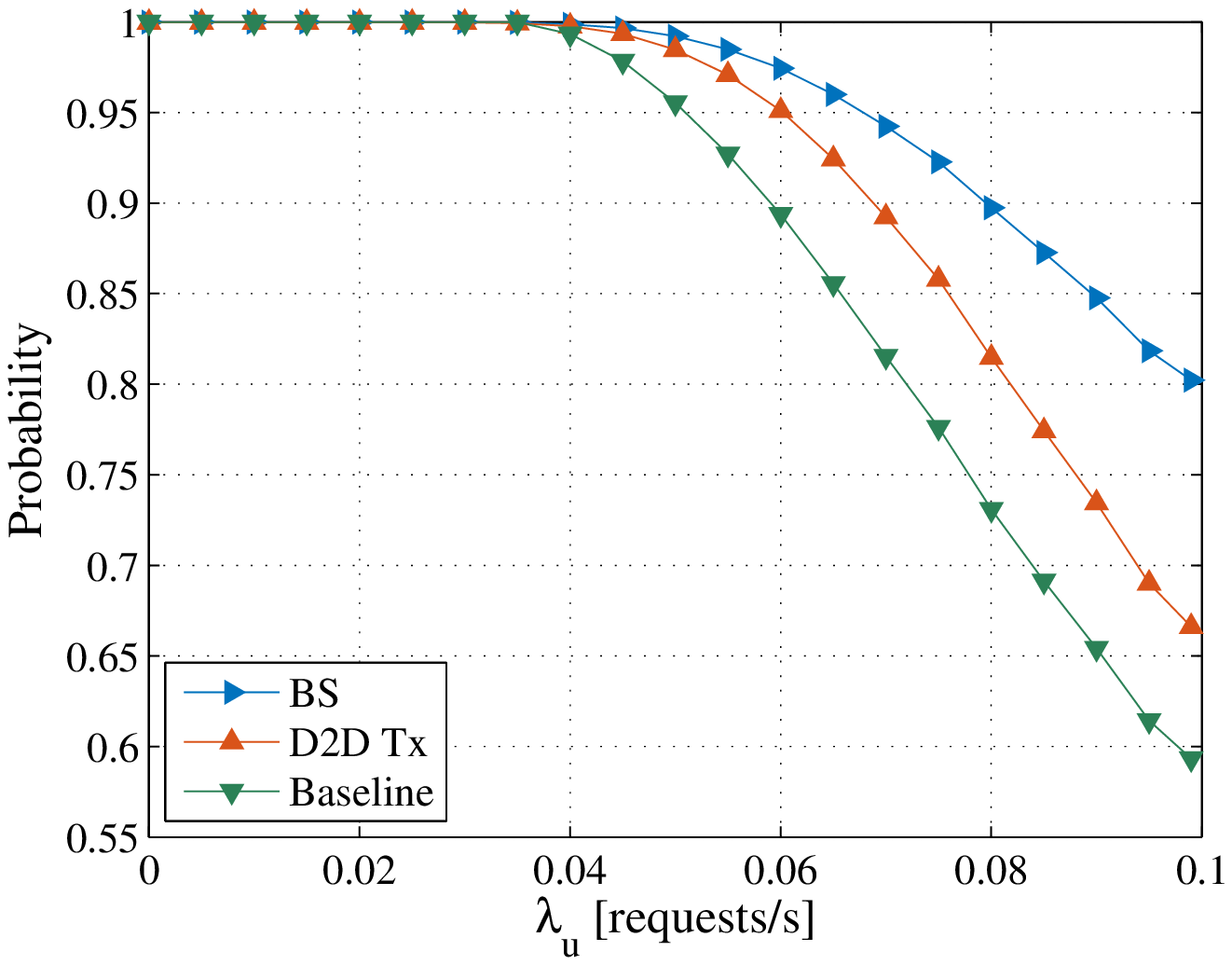}}
  \hspace{-0.3in}
  \subfigure[]{
    \label{fig:subfig:b} 
    \includegraphics[width=1.8in]{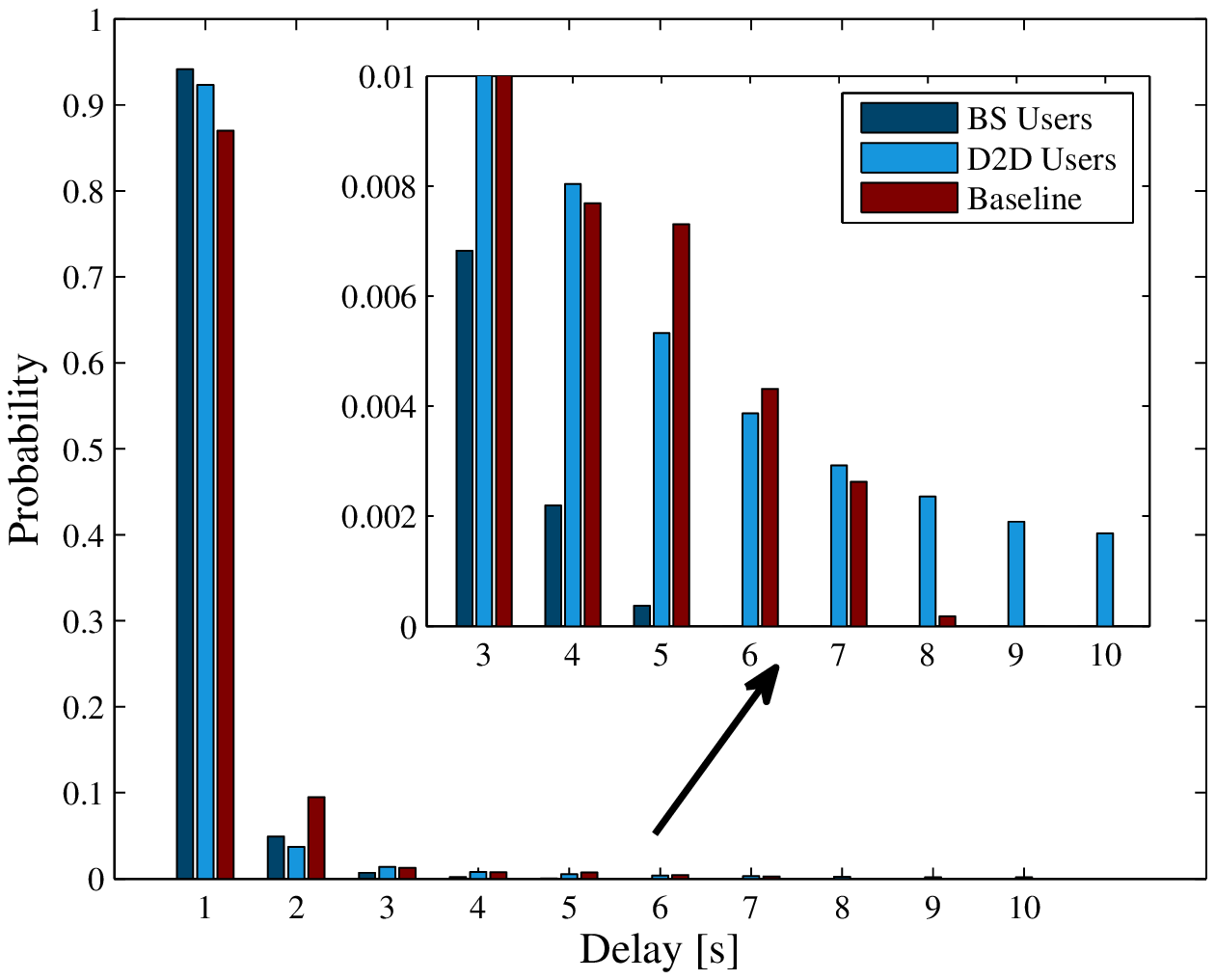}}
\caption{{{Fig.  \ref{fig:subfig:a}: The fraction of the  steady-state service nodes; Fig. \ref{fig:subfig:b}: The PMF of the request delay at the steady service nodes.}}}
  \label{fig:subfig} 
\end{figure}

Fig. \ref{fig:subfig:a} shows the fraction of the steady-state service nodes (BSs or D2D TXs) for the cognitive and cache-enabled network, in comparison with that of the baseline. For a service node, the arrival rate of requests is decided by the number of users in its cell. Due to the stochastic locations of  both users and service nodes, there is possibility that a service node has large number of requests in its cell, yielding congestion. As shown in Fig. \ref{fig:subfig:a}, nearly all service nodes stay at the steady state with lower request intensity. As to the heavy load situation, the fraction of the steady-state service nodes decreases sharply for both the baseline and the proposed network. Moreover, in the proposed network, D2D TXs  are more likely to suffer from local unsteady state because of the lower priority. However, both kinds of service nodes  (BSs or D2D TXs)  in the proposed system can bear heavier traffic load than BSs in the baseline. When the request rate is $0.1$ $[$requests/s$]$, the fraction of the steady  BSs  and D2D TXs are, respectively, $33.3\%$  and $9\%$ higher than that of the steady BSs in the baseline, which illustrates the appealing advantage of the proposed network.

The PMF of the request delay at the steady service nodes  (requests at the steady BSs and the steady D2D TXs are considered as two kinds) are illustrated in Fig. \ref{fig:subfig:b}.  The user under the coverage of steady BSs  in the proposed network has smaller delay compared with the D2D users and the users in the baseline. Furthermore, the delay of users under steady D2D has the distribution with longer tail and larger average than that under the steady BS due to the low priority. However, both kinds of users under steady service nodes  (BSs or D2D TXs) in the proposed network have better performance than that under the steady BSs in the baseline.
\section{Conclusion}
The paper aims to formulate and evaluate the performance of the multi-channel network where the D2D and cellular coexist.  We advocate to pre-cache the most popular contents to cache-enabled user via broadcasting when the network is off-peak, for the future frequent access. A cognitive D2D TX senses the idle cellular spectrum within its spectrum sensing region. Firstly, we model the node locations as mutually independent PPPs. Users dynamically  connect to the cellular and the D2D according to the jointly physical layer and content-centric features. Users are classified into three Subsets according to where the user can obtain the requested content.  Then the request arrivals and departures at the
BS and D2D TX are modeled as the \emph{discrete-time multiserver queue with priorities}. We further provide the PMF of the delay and the queue
length. Numerical results reveal that the fraction of the steady  BSs in  the proposed system is $33.3\%$  higher than that in the baseline.
\newpage
\bibliographystyle{IEEEtran}
\bibliography{paper}
\end{document}